\documentclass[a4paper,11pt]{article} 
\usepackage[T1]{fontenc}                
\usepackage[utf8]{inputenc}             
\usepackage{amssymb,amsthm,amsmath}    
\usepackage{enumerate}
\usepackage{paralist}
\usepackage{authblk}
\usepackage{graphicx}
\usepackage{enumerate}
\usepackage{array}
\usepackage{tikz}
\usetikzlibrary{calc}
\usetikzlibrary{decorations.pathreplacing}

\newtheorem{theorem}{Theorem}
\newtheorem{lemma}{Lemma}
\newtheorem{corollary}{Corollary}
\newtheorem{prob}{Problem}
\newtheorem*{prob*}{Problem}

\newtheorem*{prop*}{Proposition}

\newtheorem*{observ*}{Observation}
\newcommand{\EXTRA}[1]{}
\newcommand{\eps}{\varepsilon}
\newcommand{\Z}{\mathbb Z}
\newcommand{\N}{\mathbb N}

\newcommand{\M}{\mathcal M}

\newcommand{\e}{\varepsilon}
\newcommand{\ov}[1]{\overline{#1}}

\newcommand{\B}{\mathalpha{\star}}    

\newcommand{\derive}[1]{\xrightarrow{#1}}
\newcommand{\zps}{\mathbb Z\text{PCP}_s}
\newcommand{\ops}{\omega\text{PCP}_s}
\newcommand{\fsps}{\text{PCP}_s}
\newcommand{\cv}[2]{\underset{#2}{#1}}

\newcommand{\halt}{H}
\newcommand{\set}[1]{\left\{#1\right\}}
\newcommand{\w}[2]{#1_{1}\cdots #1_{#2}}

\newcommand{\om}{\omega}
\newcommand{\Si}{\mathsf\Sigma}

\newcommand{\Ga}{\mathsf\Gamma}

\newcommand{\fa}{\forall}

\title{On the Complexity of the Bi-infinite Post Correspondence Problem}
\author{Olivier Finkel} 
\affil[1]{Institut de Math\'ematiques de Jussieu - Paris Rive Gauche\\
CNRS, Universit\'e Paris Cit\'e,  Sorbonne Universit\'e, Paris, France. \textit{finkel@math.univ-paris-diderot.fr}}
\author[2]{Vesa Halava}
\affil[2]{Department of Mathematics and Statistics, University of Turku, Finland. \textit{vesa.halava@utu.fi }} 

\begin{document}

\maketitle

\begin{abstract}

In the bi-infinite Post Correspondence Problem ($\Z$PCP), it is asked whether the same bi-infinite word can be constructed correspondingly from a given finite set of pairs of words. In this article, we study its complexity with respect to the arithmetical hierarchy and prove that it is in $\Si^0_2 \setminus (\Pi^0_1 \cup \Si^0_1)$ and, therefore, at the level 2 of the arithmetical hierarchy. For the proof, we present a sequence of reductions starting from the nonhalting of the Turing machine all the way to $\Z$PCP via infinite PCP, an $s$-shift infinite PCP and $s$-shift $\Z$PCP for all natural numbers $s$. In the process, we prove that the infinite PCP is undecidable for injective morphisms, and that the infinite injective PCP, $s$-shift infinite PCP, $s$-shift $\Z$PCP and the non-termination problem for (deterministic and reversible) semi-Thue systems are all $\Pi^0_1$-complete.

\vspace{1em}

\end{abstract}

\section{Introduction}
In 1946 Emil Post defined and proved undecidable a problem that is today called \emph{Post Correspondence Problem} (PCP) in \cite{Po-46}. Instead of the original definition of the PCP with lists of words by Post, we shall give a definition using word morphisms:
\EXTRA{ 
with the following formulation:

\begin{prob}[Original PCP]\label{PCPlist}
Let $B$ be a finite alphabet. Given a finite set  of pairs of words over $B$, say $\{(u_1,v_1),(u_2,v_2),  \dots,$ $(u_n,v_n)\}$,  does there exist a nonempty finite sequence $i_1,\dots,i_k$ of indices such that 
\begin{equation}\label{solution_list}
u_{i_1}u_{i_2}\cdots u_{i_k} =v_{i_1}v_{i_2}\cdots v_{i_k}\, ?
\end{equation}
\end{prob}
We call the set of pairs $\{(u_i,v_i)\mid i=1,\dots, n \}$ an \emph{instance} of the PCP and a sequence with a property of \eqref{solution} a \emph{solution}. Note carefully that the alphabet $B$ in the above definition can be assumed binary (and actually already Post assumed so).    
The PCP can be expressed in terms of word monoids (also using morphisms) in the following way: recall that
}
Let  $A^*$ denote the set of all finite words over the alphabet $A$. A mapping $g\colon A^* \to B^*$ is a \emph{morphism} if $g(uv)=g(u)g(v)$ for all $u,v\in A^*$.

\begin{prob}[PCP]
Let $A=\{a_1,\dots, a_n\}$ be a finite alphabet. Given two morphisms $g,h\colon A^*\to B^*$,  does there exists a nonempty word $w=a_{i_1}\cdots a_{i_k}\in A^*$ such that 
\begin{equation}\label{solution}
g(w) =h(w)\, ?
\end{equation}
\end{prob}
The pair $(g,h)$ is called an \emph{instance} of the PCP, and the non-empty word $w$ is called a \emph{solution} of the instance $(g,h)$. Moreover, the \emph{size} of an instance $(g,h)$ of the PCP is the size $n$ of the alphabet $A$. Note carefully that the alphabet $B$ in the above definition can be assumed binary (and actually Post already defined the PCP with binary words).

The PCP and its many variants imply a bridge from combinatorial undecidable problems of computational systems and formal rewriting systems to decision problems in algebraic settings, and, therefore,  play an important role in the Theory of Computability. Also, the boundary between decidability and undecidability is revealed by the PCP and its variants. For example, it is known that the PCP is decidable for $n=2$, see \cite{EKR82,HHH02}, and undecidable for $n=5$, see \cite{Neary}. 
 
In this article we concentrate on variants of the PCP where instead of a finite sequence of indices we are asked for infinite or even bi-infinite sequences of indices such that the words built from the words of the corresponding pairs are equal.  Let us begin with the \emph{infinite Post Correspondence Problem} ($\omega$PCP):
\begin{prob}[$\omega$PCP]
Let $A=\{a_1,\dots, a_n\}$ and $B$ be two finite alphabets. Given two morphisms $g,h\colon A^*\to B^*$,  does there exist an infinite word $w=a_{i_1}a_{i_2}\dots$ such that 
\begin{equation*}
g(w) =h(w)\, ?
\end{equation*}
\end{prob}

It is known that the $\omega$PCP is  undecidable for $n=8$, see \cite{dongliu}, and decidable for instances with $n=2$, see \cite{HHK}. It is also known that the $\omega$PCP is not ''more complex'' than the PCP itself with respect to the arithmetical hierarchy as $\omega$PCP is $\Pi_1^0$-complete and the PCP is $\Sigma_1^0$-complete, see \cite{Finkel}. Indeed, this shows that $\omega$PCP is on the first level of the arithmetical hierarchy, as the PCP is.  

Our main motivation for this article is the complexity of the \emph{bi-infinite Post Correspondence Problem} ($\Z$PCP): 
\begin{prob}[$\Z$PCP]
Let $A=\{a_1,\dots, a_n\}$ and $B$ be two  finite alphabets. Given two morphisms $g,h\colon A^*\to B^*$,  does there exist a bi-infinite word  $\alpha= \dots a_{i_{-2}} a_{i_{-1}} a_{i_0} a_{i_1}a_{i_2}\dots$ such that 
\begin{equation}\label{bisol}
g(\alpha)=h(\alpha) ?
\end{equation}
\end{prob}
Note that already the equality of  bi-infinite words needs to be defined properly: for a bi-infinite word $w$, denote by $w(j)$ the letter in position $j\in \Z$ in the sequence $w$. Then $w = v$ if and only if there is a constant $\tau\in \Z$ called the \emph{shift} such that the letters $w(i)=v(i+ \tau)$ for all positions $i \in \Z$. 

An \emph{instance} of the $\Z$PCP is the pair $(g,h)$ of morphisms,  and a bi-infinite word $\alpha$ satisfying \eqref{bisol} is said to be a \emph{solution} of the instance.  The $\Z$PCP was originally proved to be undecidable in \cite{Ruo}, see \cite{bi-inf} for a somewhat simpler proof. Recently, it was proved that the $\Z$PCP is in  the class $\Sigma_2^0$ of the arithmetical hierarchy, see \cite{FHHS}. Moreover, the same work proved that $\Z$PCP \(\notin\Pi^0_1\),
using the undecidability argument of \cite{bi-inf}.
Thus, in order to locate the problem strictly at the second level, it remains
to prove that it is not in \(\Sigma^0_1\); this will follow here from the
\(\Pi^0_1\)-hardness result. In this article we sharpen the result and show that $\Z$PCP is  in $\Sigma^0_2 \setminus (\Pi^0_1 \cup \Sigma^0_1)$ and, therefore, at the level 2 of the arithmetical hierarchy.  For this we shall define  new variants of both $\omega$PCP and $\Z$PCP where the shift is fixed with a given natural number and prove that these variants are $\Pi_1^0$-complete.  We shall define these variants along the proof.  Also, we shall prove that the $\omega$PCP is undecidable for injective morphims, and that this problem is also $\Pi^0_1$-complete. Actually, we shall construct the whole chain of undecidability proof from the non-halting of the Turing machines to these variants of the PCP via the non-termination problem of the semi-Thue systems, which is also proved to be $\Pi^0_1$-complete.    

\section{Preliminaries}

Let $A$ be a finite alphabet. A \emph{word} over the alphabet $A$ is a finite sequence of symbols (also called  letters) in $A$.  We use the notation $|w|$ for the \emph{length} of a word $w$ that is equal to the number of letters in $w$.  We define the indexing of the letters of a word $v$ so that $v=v(0)\cdots v(|v|-1)$, where $v(i)\in A$ for all $i=0,\dots , |v|-1$. The \emph{empty word} is denoted by $\e$ and its length is 0. As usual, we denote the \emph{mirror image} (or the \emph{reversal}) of a word $u$ by $u^R$. The set of all words over $A$ is denoted by $A^*$.

If $v=uw$ for some word $w$, then $u$ is a \emph{prefix} of $v$, denoted by $u\le_p v$. The prefix $u$ is called \emph{proper}, denoted by $u<_p v$ if $w$ is non-empty. Similarly, if $v=wu$ for some word $w$, then $u$ is a \emph{suffix} of $v$, denoted by $u\le_s v$. The suffix $u$ is called \emph{proper}, denoted by $u<_s v$ if $w$ is non-empty. If   $v=wut$ for some words $w$ and $t$, then $u$ is a \emph{factor} of $v$. We say that words $u$ and  $v$ are \emph{comparable}, if $u\le_p v$ or $v\le_p u$.

An \emph{infinite word} $v$ (or $\omega$-word) over the alphabet $A$  is a right-infinite sequence of symbols of $A$, denoted $v=v(0)v(1)\cdots$, where the $v_i$ are letters from $A$.  We denote by $A^\omega$ the set of all infinite words over $A$.  For $v\in A^\omega$ and $n\geq 0$ an integer, we denote $v[n]=v(0)\cdots v(n)$. A \emph{bi-infinite word} over the alphabet $A$ is a two-way infinite sequence of symbols of $A$. We denote the set of bi-infinite words over $A$ by $A^\Z$. For a $w\in A^\Z$, there is an indexing such that $w=\cdots w(-2)w(-1)w(0)w(1)w(2)\cdots $ such that $w(i)\in A$ for all $i\in \Z$. Note that this indexing is not unique. Therefore, we say that two bi-infinite words $w$ and $v$ are equal, that is $w=v$ if there exists a (constant) \emph{shift} $s\in\Z$ such that $w(i)=v(i+s)$ for all $i\in \Z$. Note carefully that for image $h(w)$ of a bi-infinite word $w=\cdots w(-2)w(-1)w(0)w(1)w(2)\cdots $ under a morphisms $h$, we assume that the indexing satisfies the following condition: 
\[
\text{ in } h(w) \text{ the position 0 is the position of } h(w(0))(0), 
\]
that is, in the image the letter with index zero corresponds to the first letter of the image of the letter with index zero in $w$.

Our proofs use several constructions and methods known for the special kind of rewriting systems called semi-Thue systems.  A \emph{semi-Thue system} $T$ is a pair  $(\Gamma,R)$ where  $\Gamma= \{ a_1, a_2, \dots , a_n\}$ is a finite alphabet,  the elements of which are called \emph{generators} of~$T$, and the relation $R \subseteq \Gamma^*
\times \Gamma^*$ is  the set of \emph{rules} of $T$. We write $u \derive{}_{T} v$, if there exists a rule $(x,y) \in R$ such that $u=u_1xu_2$ and $v=u_1yu_2$ for some words $u_1,u_2\in \Gamma^*$. We denote by $\derive{}^*_{T}$ the reflexive and transitive closure of $\derive{}_{T} $, and by $\derive{}^+_{T}$ the transitive closure of $\derive{}_{T} $. Note that the index $T$ is omitted from the notation, when the semi-Thue system studied is clear from the context. If $u\derive{}^* v$ in $T$, we say that there is a derivation from  $u$ to $v$ in $T$. Also, the rules $(x,y)\in R$ are sometimes denoted by $x\derive{} y$. 

\begin{prob}[Non-Termination]
Given a semi-Thue system $T=(\Gamma,R)$ and a word $w\in \Gamma^*$, does there exists an infinite derivation starting from $w$ in $T$. 
\end{prob}
Indeed, we say that a semi-Thue system $T$ \emph{terminates} on $w$ if all derivations starting from $w$ in $T$ are finite.

\EXTRA{
In the \emph{word problem} for semi-Thue systems, is it asked, for a given semi-Thue system $T$ and words $w$ and $u$, whether $w\derive{}_T^* u$. 
}

Our sequence of reductions begins from the classical Turing machines and so called halting problem.  

\EXTRA{A \emph{Turing  machine}, TM for short,  $\mathcal{M}$ is of the form $\mathcal{M}=(Q, \Si, \Ga, \delta, q_0, F)$, where $Q$ is a finite set of states, $\Si$ is a finite input alphabet, $\Ga$ is a finite tape alphabet satisfying $\Si  \subseteq \Ga$, containing a special blank symbol $\Box \in \Ga \setminus \Si$,  $q_0$ is a unique \emph{initial state},  $\delta$ is a \emph{transition} mapping from $Q \times \Ga$ to subsets of $Q \times \Ga \times \{L, R, S\}$, and $F\subseteq Q$ is the set of \emph{accepting states}. We assume that in a TM, the transition mapping is a partial function as the TM's are assumed to be deterministic.}     

A \emph{Turing machine}, TM for short, is a 7--tuple
\[
\M=(Q,{\Sigma},\Gamma,\delta,q_0,\B, F)\,,
\]
where $Q$ is a finite set of states,  $q_0$ is the initial state,  $F\subset Q$
is the set of halting states, ${\Sigma}$ is the input alphabet, $\Gamma$ is the tape
alphabet with ${\Sigma}\subseteq  \Gamma$,  and $\delta$ is a partial function
 $Q\times\Gamma\to Q\times\Gamma\times\{L,R\}$ called the \emph{transition function}
 where $L$ and $R$ are special direction symbols and
$\B\in\Gamma$ is the blank symbol. 

Note that the TM  $\M$ is deterministic and that we allow $\delta$ to be  a partial function. 

Each \emph{transition} of $\M$ is of the
form $\delta(p,a) = (q,b,D)$ where $D$ refers to the direction of the move, that is, values $L$ and $R$ refer to left move
and right move, respectively.

A configuration of a TM $\M$ in its computation is defined as a word
$u(q,a)v\in\Gamma^*(Q\times \Gamma)\Gamma^+$ where $q \in Q$ and $u =
\eps$ or $u$ begins with a nonblank letter, $v = \B$ or $v$ ends with a nonblank
letter. A configuration corresponds to the global state of the Turing machine where $uav$ is the shortest content of the tape containing the square pointed by the read-write head and all nonblank symbols of tape and  $\M$ is reading the symbol
$a$ on its position in state $q$.

A \emph{step} in a computation $\gamma \vdash_\M \gamma'$
yielding from a configuration $\gamma$ to the next one $\gamma'$ in $\M$ is
defined as follows. Let $\gamma=a_1a_2 \cdots a_{i-1} (p, a_i) a_{i+1}\cdots a_n$. Note that if $i=n-1$, then it is possible that $a_n=\B$, and similarly, if $i=1$, then it is possible that $a_i=\B$. 

\smallskip
\noindent
(L) Consider a left transition $\delta(p,a_i)=(q,b,L)$. Then
\[
\gamma \  \vdash_\M \
\begin{cases}
 \w{a}{i-2}(q,a_{i-1})ba_{i+1}\cdots a_n &\text{if }1 < i \text{ and }a_n\ne \B,\\
 (q,\B) ba_{i+1}\cdots a_n & \text{if }i=1,\\
\w{a}{i-2}(q,a_{i-1})b &\text{if }i=n-1 \text{ and } a_n=\B\,. \end{cases}
 \]

\noindent
(R) For a right transition $\delta(p,a_i)=(q,b,R)$, let
\[
\gamma \ \vdash_\M \
\begin{cases}
\w{a}{i-1}b(q,a_{i+1})\cdots a_n &\text{if }i<n-1 \text{ and } a_1\ne\B\\
 \w{a}{i-1}b(q,a_{i+1})\B  &\text{if }i=n-1,\\
b(q,a_{i+1})a_{i+2}\cdots a_n & \text{if }i=1 \text{ and } a_1=\B\,.\end{cases}
\]

We make three general assumptions on TM's. We assume that  
\begin{enumerate}
\item[(M1)]  $F=\set{\halt}$, so there is a unique halting state, 
\item[(M2)] for all $a\in \Gamma$,  $\delta(\halt,a)$ is not defined, that is, when the state $\halt$ is reached, the TM halts,  
\item[(M3)] TM never returns to the initial state $q_0$ after the first move,

\item[(M4)] for all $a\in \Gamma$ and $q\in Q\setminus \set{q_0, \halt}$, $\delta(q,a)$ is defined, and 
\item[(M5)] TM never writes $\B$ to the tape.

\end{enumerate}
All five assumptions can be fulfilled without losing generality, by adding a new unique halting state and, if needed, a so called garbage state and new transitions for these, and adding a new symbol $\B'$ for written blanks on the tape  and new state $q_0'$ together with needed transitions for these to mimic transitions reading $\B$ and returning to $q_0$.
Based on this, for each configuration with $q\in Q\setminus\set{\halt}$, 
$\gamma = u(q,a)v$, there exists exactly one configuration
$\gamma'$ such that $\gamma \vdash_\M \gamma'$.

Let $\vdash_\M^*$ or \, $\vdash^*$, for short, be the reflexive and transitive
closure of the relation $\vdash_\M$. Thus $\gamma \vdash^* \gamma'$ if and only
if there exists a finite sequence
\begin{equation}\label{tm-comput}
\gamma=\gamma_1 \vdash \gamma_2 \vdash \dots \vdash \gamma_k=\gamma'
\end{equation}
of configurations for some $k \ge 1$ including the possibility that $\gamma =
\gamma'$. Such a sequence is called a \emph{computation} of $\M$. It is a
\emph{halting} computation if the state in $\gamma'$ is the unique final state $\halt$.

Our undecidability proof employs the \emph{halting problem} of Turing machines
on empty input; see, e.g., \cite{Manna}.

\begin{theorem}\label{missa}
For a given TM $\M$, it is undecidable whether or not $\M$ halts on empty input,
that is, whether or not $(q_0,\B) \vdash^* u(\halt ,a)v$, where $u,v\in
\Gamma^*$, $a\in \Gamma$.
\end{theorem}

We shall also need in the sequel the notion of a Turing machine reading infinite words. We now recall these notions. 

\EXTRA{The {\it first infinite ordinal} is $\om$.
 An $\om$-{\it word} over an alphabet $\Si$ is an $\om$-sequence $a_1a_2a_3 \cdots$, where for all 
integers $ i\geq 1$, ~
$a_i \in\Si$.     
 The {\it set of } $\om$-{\it words} over  the alphabet $\Si$ is denoted by $\Si^\om$.
An  $\om$-{\it language} over an alphabet $\Si$ is a subset of  $\Si^\om$.  For an $\om$-word $\sigma=a_1a_2a_3\cdots$, we denote the prefix $a_1\cdots a_n$ by $\sigma[n]$.
}

Assume that a {Turing 
machine} $\mathcal{M}$ is of the form $\mathcal{M}=(Q, \Si, \Ga, \delta, q_0, F)$, where $F\subseteq Q$ is the set of {accepting states}. Turing machines reading  infinite words have been considered in~\cite{CG78b,Staiger97}.
A Turing machine $\mathcal{M}$ accepts a word 
$\sigma\in \Si^\om$ with the $2$\emph{-acceptance  condition} if and only if there is an infinite run of $\mathcal{M}$ on input $\sigma$ 
visiting infinitely often states from $F$. The $2$-acceptance condition is also now known as the B\"uchi acceptance condition. 
On the other hand,  a Turing machine $\mathcal{M}$ accepts a word 
$\sigma\in \Si^\om$ with  $2'$-\emph{acceptance  condition} if and only if there is an infinite run of $\mathcal{M}$ on $\sigma$ 
visiting only finitely often the accepting states in  $F$. The $2'$-acceptance condition is also now known as the co-B\"uchi acceptance condition.  Moreover, a Turing machine $\mathcal{M}$ reading infinite words over a finite alphabet $\Si$ accepts a word 
$\sigma\in \Si^\om$ for $1'$-acceptance  condition iff there is an infinite run of $\mathcal{M}$ on $\sigma$ 
visiting only states in the set $F \subseteq Q$ of accepting states.  

We require in this article that an accepting  run should be infinite on the input $\sigma\in \Si^\omega$, as in \cite{Staiger97}, and not that it is complete (i.e. we do not require that all the cells of the right-infinite tape of the Turing machine are visited nor that all letters of $\sigma$ are read), or even non-oscillating, as in \cite{CG78b}. We refer the interested reader to \cite{Fin-ambTM} and papers cited in \cite{Staiger97,Fin-ambTM}  for a comparison between these modes of acceptance of infinite words by Turing machines.

We assume the reader to be familiar with the arithmetical hierarchy on subsets of  $\mathbb{N}$, as a general reference we give \cite{rog,Odifreddi1}. We now recall the definition of  the arithmetical   hierarchy   on subsets of  $\Si^\om$ for a finite alphabet $\Si$, see \cite{Staiger97}. 
 An $\om$-language $L\subseteq \Si^\om$  belongs to the class  $\Sigma^0_n$ if and only if there exists a recursive relation  $R_L\subseteq \mathbb{N}^{n-1}\times \Si^\star$  such that
\[
L = \{\sigma \in \Si^\om \mid \exists x_1 Q_2 x_2\ldots Q_n x_n  \quad (x_1,\ldots , x_{n-1}, 
\sigma[x_n+1])\in R_L \}, 
\]
where $Q_i$  for $i=2,\dots, n$ is one of the quantifiers $\fa$ or $\exists$  (not necessarily in an alternating order). An $\om$-language  $L\subseteq \Si^\om$  belongs to the class $\Pi^0_n$ iff its complement $\Si^\om - L$  belongs to the class  $\Sigma^0_n$. The inclusion relations that hold  between the classes $\Sigma^0_n$ and $\Pi^0_n$ are  the same as for the corresponding classes of the Borel hierarchy. The classes $\Sigma^0_n$ and $\Pi^0_n$ are strictly included in the respective classes  ${\bf \Sigma}_n^0$ and ${\bf \Pi}_n^0$ of the Borel hierarchy. 

An important result is that the modes of acceptance of  $\om$-languages by deterministic  Turing machines are connected  to the classes of the  arithmetical   hierarchy.  In particular,  an $\om$-language is in the   arithmetical  class $\Pi^0_2$ (respectively, $\Sigma^0_2$) if and only if it is accepted by a {\it deterministic } Turing machine with $2$-acceptance condition, i.e. B\"uchi acceptance condition (respectively, with $2'$-acceptance condition, i.e., co-B\"uchi acceptance condition), see Corollary 2.3   in \cite{Staiger97}.

\section{From non-halting of TM's to non-termination of reversible semi-Thue systems }\label{sec:revsT}

In this section we shall present a variant of the reduction in \cite{HalHarSF} (see also \cite{HaHa-13}) from deterministic Turing machines to deterministic and reversible semi-Thue systems. The original idea of the construction is based on the construction in \cite{KS10}, and we shall use a version closer to this construction. See also \cite{HHK:PCP} for the details.

We need to start our sequence of reductions from the Turing machines to get the complexity results correct.  Note also, that the reversibility property mentioned above is employed only in the final reduction of this article in  Sec.~\ref{finalred} for the complexity of the $\Z$PCP. 

To obtain injectivity, we modify the construction in \cite{HalHarSF}.

Consider a semi-Thue system $T=(A,R)$. Let $B$ be a nonempty subset of $A$, and define
\[
s_A(B) = (A \setminus B)^* B (A \setminus B)^*
\]
be the set of words over $A$ that contain a single occurrence from the set $B$.
In \cite{HalHarSF}, semi-Thue system $T$ was defined to be $B$-\emph{deterministic} for a nonempty subset $B \subseteq A$, if each rule $(u,v) \in R$ belongs to $s_A(B) \times s_A(B)$, and
for all words $w \in s_A(B)$, there exists at most one derivation from $w$, i.e.,
if $w \derive{}_T u$ and $w \derive{}_T v$, then $u=v$. Further, $T$ is called $B$-\emph{reversible}, if
for all $u_1, u_2 \in s_A(B)$,
$u_1 \derive{}_T v$ and $u_2 \derive{}_T v$ imply $u_1=u_2$ and the
rule applied in these steps is the same.

Let $\M=(Q,{\Sigma},\Gamma,\delta,q_0,\B,\halt) $
be a given Turing machine satisfying the assumptions (M1)--(M5). We define three alphabets. Let
$$
\Delta = \Gamma \cup (Q \times \Gamma)
$$
and
$$
\Theta = \Delta
\times (\{ S \} \cup (Q \times \Gamma)),
$$
where $S$ is a new symbol. We write the second component of an element of
$\Theta$ below the first component. For example, if $q \in Q$ and $a \in
\Gamma$, then $\cv{(q,a)}{S} \in \Theta$ and $\cv{a}{(q,a)} \in \Theta$. 
Finally, let
$$
\Lambda=\Delta \cup \Theta\cup\set{\#,\$ }= (\Gamma\cup (Q\times \Gamma))\cup  (\Delta
\times (\{ S \} \cup (Q \times \Gamma)))\cup\set{\#,\$ },
$$
where $\#$ and $\$ $ are new symbols.

For TM $\M$, we define the semi-Thue system
\[
S_{\M}=(\Lambda, R)
\]
with the following idea: the special symbols~$\#$ and~$\$
$ are the endmarkers of the tape of $\M$. For the initial configuration
$(q_0,\B)$ of $\M$ we have the word $\#\cv{(q,\B)}{S}\$$.
For a symbol $\cv{(q,b)}{S}$, we simulate a step of the TM $\M$ as follows:
 \begin{itemize}\label{behaviour_SM}
     \item[(S1)] we apply the transition of $\M$ and remember the symbol $(q,b)$ in the lower part of the symbol in $\Theta$ (see rules \eqref{eq:sim1} - \eqref{eq:sim4} below). 
     \item[(S2)] Then the symbol $(q,b)$ is walked into the left marker symbol $\#$ (by the
rules \eqref{eq:left}), and 
      \item[(S3)]stored on the left-hand side of it (by the rule
\eqref{eq:forget}). 
\item[(S4)] Then the lower part of symbol in $\Delta\times (\set{S}\cup(Q\times \Gamma))$ is changed to $S$ and  it is walked back to the  configuration position, that is, to  the position with symbol $(q,b)\in Q\times \Gamma$ on the tape (using the rules
\eqref{eq:right}). 
 \end{itemize} 
Note that for writing on the left and right boundary, we need special rules (\eqref{eq:sim2} and \eqref{eq:sim4}), especially adding $\B$ if needed.

The rules in $R$ are defined as follows: Assume that
\begin{align*}
&b,b' \in \Gamma,  \ q\in Q, \  a \in \Gamma\cup\set{\#},  \ c\in  \Gamma\cup\set{\$},\\
&a,c,b' \ne \B, \  q'\ne q_0,\   \text{ and } \  y\in \Delta,  \ z\in \Delta\cup \set{\$}.
\end{align*}
The rules in $R$ are
\begin{align}
  a\cv{b}{S} z
      &\derive{}  a b\cv{z}{S}\, ,&z\notin\set{\$,\B, (q_0,\B)},\ \label{eq:right}\\
    a \cv{(q,b)}{S}c
       &\derive{} \cv{(q',a)}{(q,b)} b'c\,,
       &a\ne \#,\text{ if $\delta(q,b) = (q',b',L)$},\label{eq:sim1}\\
    \# \cv{(q,b)}{S}c
       &\derive{} \# \cv{(q',\B)}{(q,b)} b'c\,,
        &\text{if $\delta(q,b) = (q',b',L)$},\label{eq:sim2}\\
    a\cv{(q,b)}{S} c
        &\derive{}  a b' \cv{(q',c)}{(q,b)}\,,
        &c\ne\$,\text{ if $\delta(q,b) = (q',b',R)$},\label{eq:sim3}\\
    a\cv{(q,b)}{S} \$
        &\derive{} a b' \cv{(q',\B)}{(q,b)} \$\,,
        &\text{if $\delta(q,b) = (q',b',R)$},\label{eq:sim4}\\
   a \cv{y}{(q,b)}z\,,
        &\derive{}  \cv{a}{(q,b)}yz\, ,&a\ne \#\,, \label{eq:left}\\
    \# \cv{y}{(q,b)}z\,,
        &\derive{} (q,b) \# \cv{y}{S}z\,.&\label{eq:forget}
\end{align}

The difference with construction in \cite{HalHarSF} is that we use here triplets in the rules for the injectivity of the constructed morphisms. In \cite{HHK:PCP} the semi-Thue system $S_\M$ was extended to be able to prove that the word problem for $B$-deterministic and $B$-reversible semi-Thue systems is undecidable. As we are interested in (non-)termination here, we may continue with $S_\M$. Indeed, the following lemmata follow from results proved for the extended semi-Thue system $\ov{S}_\M$ in \cite{HHK:PCP}. The proofs are similar to the proofs in \cite{HalHarSF}.   

\begin{lemma}\label{inj:lemma1}
Let $\M$ be a given Turing machine. Then the semi-Thue system $S_{\M}$ is
$\Theta$-deterministic and $\Theta$-reversible.
\end{lemma}

\begin{proof}
Clearly, both sides of rules have words from $s_{\Lambda}(\Theta)$  for all rules in $R$.  The uniqueness of the derivation for $\Theta$-determinism follows from the determinism of TM $\M$.

Also, the $\Theta$-determinism guarantees the $\Theta$-reversibility.  All  words in $s_{\Lambda}(\Theta)$ contain at most one of the right-hand sides of the rules in $R$ as a factor.                                        
\end{proof}

\begin{lemma} \label{lem:eid}
Let  $\alpha_0=(q_0,\B) \vdash \alpha_1\vdash \alpha_2\vdash \dots$ be a
computation of  the TM $\M$, where $\alpha_i= u_i x_i v_i$ with $u_i, v_i \in
\Gamma^*$ and $x_i \in Q \times \Gamma$. Then for the (unique) derivation
$$\beta_0=\# \cv{(q_0,\B)}{S} \$ \derive{}_{S_{\M}} \beta_1\derive{}_{S_{\M}}
\beta_2 \derive{}_{S_{\M}} \dots ,
$$
there are indices $0 = j_0 < j_1 < j_2 < \dots$ such that
\begin{equation} \label{eq:eid}
    \beta_{j_i} =
    x_0 \dots x_{i-1} \# u_i \cv{x_i}{S} v_i \$\,.
\end{equation}
\end{lemma}

\begin{proof}
The proof is by induction on $i$. First, the case for $i=0$ with $j_0=0$ is clear as $\beta_0$ is of the required form.
Assume next that the claim holds for $i$. We prove that
\begin{equation}
    \beta_{j_i}= x_0 \dots x_{i-1} \# u_i \cv{x_i}{S} v_i\$  \derive{}^+_{S_{\M}}
    x_0 \cdots x_i \# u_{i+1} \cv{x_{i+1}}{S} v_{i+1} \$\,.
\end{equation}
Applying the appropriate rule \eqref{eq:sim1} - \eqref{eq:sim4} gives us that 
\begin{equation}
    \beta_{j_i+1} =
    x_0 \cdots x_{i-1} \# u_{i+1} \cv{x_{i+1}}{x_i} v_{i+1} \$
\end{equation}
and then using the rule \eqref{eq:left} sufficiently many times, and after that the rule
\eqref{eq:forget}, we obtain
\begin{equation}
  \beta_{j_i+1}\derive{}^*_{S_{\M}}  x_0 \cdots x_i \# \cv{y}{S} v \$\,,
\end{equation}
where $y v =u_{i+1} x_{i+1} v_{i+1} = \alpha_{i+1}$. Applying the rule \eqref{eq:right} sufficiently
many times, we move the symbol $S$ under the symbol $x_{i+1}\in Q\times
\Gamma$ in $yv$. Therefore there exists an index $j_{i+1}$ such that
$\beta_{j_{i+1}}$ is of the form \eqref{eq:eid}. 
\end{proof}

Next theorem follows from Theorem~\ref{missa} and the previous lemma. 

\begin{theorem}\label{TM:sT}
The non-termination problem is undecidable for semi-Thue systems. Indeed, TM $\M$ does not halt on the empty input if and only if $S_\M$ is non-terminating for the initial word $\# \cv{(q_0,\B)}{S} \$$.     
\end{theorem}

Note that if $\M$ halts, then it halts on configuration 
\[
u(\halt,a)v, \quad \text{ for some }a\in \Gamma. 
\]
As there is no transition in $\M$ for $(\halt,a)$, in $S_\M$ the last word of the derivation from the initial word is $\# \cv{(q_0,\B)}{S} \$$ is of the form 
\[
 x_0 \dots x_{n}\# u \cv{(\halt,a)}{S} v \$ .
\]

\begin{theorem}\label{thm:stcomp}
The non-termination problem for semi-Thue systems is $\Pi^0_1$-complete. Moreover, the termination problem for semi-Thue systems is $\Sigma^0_1$-complete.

\end{theorem}
\begin{proof}
It is well known that the non-halting problem of  deterministic Turing machines is $\Pi^0_1$-complete. Therefore, the above construction gives that the non-termination problem for semi-Thue systems is $\Pi^0_1$-hard. 

To show that it is $\Pi^0_1$-complete we still need to show that it is in $\Pi^0_1$.

For that, we follow the steps of the construction in~\cite{Finkel} for the complexity of the $\omega$PCP. 

 Let now an instance ${\rm Ins}=((A,R), w)$ of the non-termination problem for semi-Thue systems be given by a semi-Thue system $(A,R)$ and a finite word $w\in A^\star$. Then we can easily code infinite derivations of this semi-Thue system starting from the word $w$ by infinite words over the  finite alphabet $X=A \cup \{\derive{} \}$, hence also over the alphabet $\{0, 1\}$. 
 
 Then we can easily associate, in a recursive manner, to this instance ${\rm Ins}=((A,R), w)$ a deterministic Turing machine $\mathcal{M}$ with  $1'$-acceptance  condition
  reading infinite words over the alphabet  $\{0, 1\}$ and accepting an infinite word 
$x\in \{0, 1\}^\om$ if and only if  $x$ is (the code of) a solution of the  instance $\rm{Ins}$ of the non-termination problem for semi-Thue systems. 

 On the other hand,  the set of infinite words accepted by such a deterministic  Turing machine with  $1'$-acceptance  condition is known to be an effective $\Pi_1^0$-subset of $\{0, 1\}^\om$ (see~\cite{Staiger97}). Thus the set of solutions of the non-termination problem of instance  $\rm{Ins}$ is an effective $\Pi_1^0$-set accepted by a 
deterministic Turing machine $\mathcal{M}$ with  $1'$-acceptance  condition which can be constructed from $\rm{Ins}$. 

 Staiger proved in \cite[p. 638]{Staiger93} that it is $\Pi_1^0$-complete to decide whether such a 
$\Pi_1^0$-set is non-empty. This is also stated by Cenzer and Remmel in \cite[Theorem 4.1 (ii)]{CenzerRemmel03}. 

Therefore,  the problem to determine whether a given instance $\rm{Ins}$ of  the non-termination problem for semi-Thue system has a solution 
is in the class $\Pi_1^0$. This ends the proof. 

The second claim follows from the first one. 
\end{proof}

Note that the previous two theorems could have also been stated for the special case of $\Theta$-deterministic and $\Theta$-reversible semi-Thue systems. These stronger special properties of semi-Thue system $S_{\mathcal{M}}$ will become useful later in our constructions.

For that purpose, we prove the following lemma. 
\begin{lemma}\label{no:inf_rev_derive}
Let $\beta\in (Q\times \Gamma)^*\# \Delta^* \Theta \Delta^*\$ $. Then there is no infinite reverse derivation starting from $\beta$ with the rules of $S_\M$
\end{lemma}

\begin{proof}
Clearly, by the fact that $S_\M$ is $\Theta$-deterministic and reversible, there exists at most one word $\beta_{j_1}$ (necessarily also $\beta_{j_1}\in (Q\times \Gamma)^*\# \Delta^* \Theta \Delta^*\$ $ such that 
$\beta_{j_1}\derive{} \beta$ in $S_M$. The same holds for $\beta_{j_1}$, and therefore, there  exists at most one  sequence of words $\beta_{j_i}\in (Q\times \Gamma)^*\# \Delta^* \Theta \Delta^*\$ $ such that $\beta_{j_{i+1}}\derive{} \beta_{j_i}$.  

Set $\beta=\beta_{j_0}$ and let $\beta_{j_i}=w_{j_i}\#u_{j_i}X_{j_i}v_{j_i}\$ $ with $w_{j_i}\in (Q\times \Gamma)^*$ and $X_{j_i}\in \Theta$ for $i\in \N$. Consider the behaviour of $S_\M$ described on page~\pageref{behaviour_SM} in (S1)--S(4). In reverse derivation, the symbols in $\Theta$ travel to symbol $\#$ and then, reverse to (S3), remove a symbol from left of $\#$. Then the symbols in $\Theta$ travel to the right, cancel a transition of $\M$, then travel left to the symbol \# and again remove a symbol left from it. Therefore, in the sequence $(\beta_{j_i})$, there is a finite subsequence $(\beta_{t_i})$ with $|w_{t_{i+1}}|< |w_{t_i}|$. Now, there are only finitely many reverse transitions after the word $w$ has been totally removed, implies that the sequence $(\beta_{j_i})$ is finite. 

\end{proof}

\section{Injective infinite PCP}

We shall next construct the morphisms for the proof of injective infinite PCP. 
Let $\Delta_1 = \Delta \cup \{\#, \$\}$. Therefore, 
\[
\Lambda =\Delta_1 \cup \Theta\,.
\]
Let the semi-Thue system $S_{\M}=(\Lambda, R)$ be  defined as in the previous section. We define two more alphabets which will serve as the domain and the image alphabet for the morphisms defined next, let
\[
A=\Delta_1\cup R\cup\set{d}= \Gamma \cup (Q\times\Gamma)\cup \{\#, \$,d\}\cup R ,
\]
where $R=\{t_1,\ldots, t_n\}$ is the set of rules of
$S_{\M}$, and $d$ is a new symbol. Finally, let
\[
B=\Lambda\cup\set{d}.
\]

We shall define two morphisms
\[
g, h\colon  A^* \to B^*
\]
following the idea of construction in \cite{Cl-80}. However, instead of the binary alphabet $\{a,b\}$, the alphabet will be $A$. Note that the symbols from $\Theta$ are not in $A$ and they appear only in the images of the rules $t\in R$.

First, we define two special mappings. Let $Y$ be an alphabet and let $k$ be a non-empty word over some alphabet. Let
$\ell_k, r_k \colon Y^*\to (Y\cup \{s\})^*$
be the left and right \emph{desynchronizing morphisms} defined by
$$\ell_k(a) = k a \quad \text{ and } \quad  r_k(a) = ak$$
for all letters $a \in Y$. Clearly, $\ell_k(w)\cdot k = k\cdot r_k(w)$ for all words $w$.

Define the morphisms $g,h\colon A^*\to B^*$ by 
\begin{equation} \label{inj_oPCP}
\begin{aligned}
 & g(x)= r_d(x), &&\quad h(x)= \ell_d(x) &&\text{for } x\in \Delta_1  \\  
  & g(t)=r_d(u), &&\quad h(t) = \ell_d(v)  && \text{for } t= (u\to v) \in R,    \\     
& g(d)=d,   &&\quad  h(d) = \ell_d(\#\cv{(q_0,\B)}{S}\$).  &&     \\     \end{aligned}
\end{equation}    
The next lemma shows that the morphisms $g$ and $h$ simulate the semi-Thue
system $S_{\M}$, and consequently, they simulate the given TM $\M$.

\begin{lemma} \label{lem:sim}
Let $\beta_0=\#\cv{(q_0,\B)}{S}\$ \derive{} \beta_1 \derive{} \beta_2 \derive{} \ldots$ be a
derivation in $S_{\M}$, where $\beta_j=w_1uw_2$ and
$\beta_{j+1}=w_1vw_2$ with  $t=(u\derive{} v)\in R$. Then
\[
g(w_1tw_2) = r_d(\beta_j)\quad \text{ and }\quad h(w_1tw_2) = \ell_d(\beta_{j+1}).
\]
Furthermore, if, for some word $w\in A^*$, $r_d(\beta_j) \leq_p g(w)$ or
$\ell_d(\beta_{j+1}) \leq_p h(w)$, then $w_1tw_2 \leq_p w$.
\end{lemma}

\begin{proof}
Let
\[
\beta_j = x_1 \cdots x_m X y_1 \cdots y_n\,,
\]
where  $X \in \Theta$ and $x_i, y_i \in \Delta_1$ for each $i$. If
$r_d(\beta_j) \leq g(w)$, then $w$ must begin with
\[
w'=x_1 \cdots x_{m-1} t y_2 \cdots y_n,
\]
where $t=(x_m X y_1\derive{} v)\in R$, i.e.,  $u=x_mXy_1$. Clearly, $g(w') =
r_d(\beta_j)$ and $ h(w')=\ell_d(x_1 \cdots x_{m-1}v y_2 \cdots y_n). $
By the definition of $h$, the word $v$ in transition of
$x_m X y_1\derive{} v$ is uniquely
defined and it corresponds
to the rule of  $S_{\M}$ used in $\beta_j\to \beta_{j+1} $ so that
$h(w') = \ell_d(\beta_{j+1})$.

Finally, if $\ell_d(\beta_{j+1}) \leq_p h(w)$, then similarly it can be seen that $w$
begins with $w' = x_1 \cdots x_{m-1} t y_2 \cdots y_n$. This proves the claim.
\end{proof}

\begin{lemma} \label{lem:a}
Let $\M$ be a Turing machine and $S_{\M}$ be the corresponding semi-Thue
system. Then $S_\M$ is non-terminating for the initial word
 $\#\cv{(q_0,\B)}{S}\$ $ if and only if there exists an infinite word $w\in A^\omega$ such that
 $g(w)=h(w)$.
\end{lemma}

\begin{proof}
Assume first that $S_{\M}$ is not terminating for $\#\cv{(q_0,\B)}{S}\$$, that is, there exists 
$$
\beta_0=\#\cv{(q_0,\B)}{S}\$ \derive{}\beta_1\derive{} \beta_2 \derive{} \ldots\,,
$$
where  $\beta_j=x_j u_j y_j$ and $\beta_{j+1}=x_j
v_j y_j=x_{j+1}u_{j+1}y_{j+1}$ with $t_j=(u_j\derive{} v_j)\in R$ being the rule
applied in the derivation step  $\beta_j\derive{} \beta_{j+1}$. Now, for
$w=dx_0t_0y_0x_1t_1y_1\cdots $, we have
\begin{equation*}
\begin{split}
g(w)&=dr_d(\beta_0\beta_1\cdots ) \\
&=\ell_d(\#\cv{(q_0,\B)}{S}\$)\ell_d(\beta_1\beta_2 \cdots)=h(w)\,,
\end{split}
\end{equation*}
and hence $g(w)=h(w)$ as required.

Assume next that there is a  word $w\in A^\omega$ such that $g(w)=h(w)$, and, to the  contrary, that $S_\M$ terminates on $\#\cv{(q_0,\B)}{S}\$ $.

Necessarily, $w=dw'$ for some word $w'\in A^\omega$, since only $g(d)$ and $h(d)$ are prefix comparable. Now,
$$h(d)=\ell_d(\#\cv{(q_0,\B)}{S}\$)=\ell_d(\beta_0)\,,$$
and therefore, by
Lemma~\ref{lem:sim}, $x_0t_0y_0\le w'$, where $\beta_0=x_0 u_0 y_0$ and
$\beta_{1}=x_0 v_0 y_0=x_{1}u_{1}y_{1}$ and, moreover, $t_0=(u_0\derive{} v_0)\in R$ is
the rule applied in the first step $\beta_0\derive{} \beta_{1}$. Also,
$h(x_0t_0y_0)=\ell_d(\beta_1)$, and by Lemma~\ref{lem:sim},
$x_0t_0y_0x_1t_1y_1 \le w'$. By applying Lemma~\ref{lem:sim} iteratively, we obtain
a derivation
\[
\beta_0=\#\cv{(q_0,\B)}{S}\$ \derive{}\beta_1\derive{}^* \beta_n
\]
in ${S_{\M}}$. Finally, $g(dw')^{-1}
h(dw')=\ell_d(\beta_n)$, and, as $\M$ halts only in state $\halt$ and, if and only if $S_\M$ terminates, we must have 
\[
\beta_n=x_n\cv{(\halt,a)}{S} y_n,
\]
for some words $x_n$ and $y_n$ and $a\in  \Gamma$. Now the letter $\cv{(\halt,a)}{S}$ does not appear in images of $g$, so there is no way to cover $\ell_d(\beta_n)$ with images of $g$, a contradiction.
\end{proof}

Let $A$ and $B$ be finite alphabets. A morphism $h\colon A^* \to B^*$ is said to be of \emph{bounded delay $n$}, if there are no different letters $a,b \in A$ such that $h(au) \leq_p h(bv)$ for some words $u,v \in A^n$ of length $n$. Moreover, a morphism $h$ is of \emph{bounded delay} if it is so for some $n$.

Now, it can be proved that the morphisms $g$ and $h$ are of $2$-\emph{bounded delay} and, therefore, injective.

\begin{lemma} \label{lem:b}
The morphisms $g$ and $h$ are of bounded delay $2$.
\end{lemma}

\begin{proof}
If a morphism $f$ is not of bounded delay, then, necessarily, for two different letters $r,s$,
either $f(r) = f(s)$ or $f(r) <_p f(s)$. We show that if $f=g$ or $f=h$, then the first case will not happen, and that if
$f(r) <_p f(s)$, then there are no letters $p,q$ such that $f(rpq)$ and $f(s)$ are comparable. This implies that $f$ is of bounded delay~2.

For the morphism $g$ the claim follows straightforwardly: for all $t\in R$, $g(t)= a\theta b$, where $\theta\in \Theta$ and $a,b\in \Delta_1$.
Since $\theta\notin \Delta_1$ and the Turing machine $\M$ being
deterministic, which implies that the left hand sides of rules in \eqref{eq:right}--\eqref{eq:forget} are all different, it is immediate that $g(r) \ne g(s)$ for different letters $r$ and $s$, and if $g(r) <_p g(s)$, then there exists no letter $p$ such that $g(rp)$ and $g(s)$ are comparable.

For the morphism $h$ the proof is bit more complicated. 
First of all, for different letters $r$ and $s$, the equality $h(r)=h(s)$ is not possible based on the definition of rules in \eqref{eq:right}--\eqref{eq:forget} and the determinism of the TM $\mathcal{M}$. Note that in $h$ only the letter $d$ has symbol with $q_0$ in its image. So assume  $h(r) <_p h(s)$. There are several possible cases with $s\in R\cup \{d\}$: 

If $s=d$, then $r=\#$. Now $h(p)$ should begin with symbol $\cv{(q_0,\B)}{S}$ which is not the case in any of the right hand sides of the transitions in  \eqref{eq:right}--\eqref{eq:forget}.

Assume then that  $s\in R$. We have the following cases:

1) assume $h(s)$ contains symbol $\cv{z}{S}$ with $z\in \Delta$. Then because $h(r) <_p h(s)$, necessarily $s$ is of the form \eqref{eq:forget} with $h(s)=\ell_d((q,b)\# \cv{y}{S} z)$
where $y\in \Delta$ and $z\in \Delta\cup \{\$\}$ and $r$ is of the form \eqref{eq:right} with $h(r)=\ell_d(ab\cv{z}{S})$ where $a\in \Gamma\cup \{\#\}\setminus\{\B \}$. This is a contradiction as $(q,b)\notin \Gamma\cup \{\#\}\setminus\{\B \}$. 

2) assume $h(s)$ contains symbol $\cv{(q',x)}{(q,b)}$ as in the rules \eqref{eq:sim1}-\eqref{eq:sim4}. As the TM $\mathcal{M}$ is deterministic, then lower part $(q,b)$ appear only in rules of the type \eqref{eq:sim1} and \eqref{eq:sim2} or 
\eqref{eq:sim3} and \eqref{eq:sim4}.

Obviously, $s$ cannot be of type \eqref{eq:sim1}, as then no $r$ with prefix property exists. Consider the case where $s$ is of type \eqref{eq:sim2} and $x=\B$. Then $r=\#$, but then there exists no letter $p$ as in rules of type \eqref{eq:sim1} $x\ne \B$.

If $s$ is of type \eqref{eq:sim3} or \eqref{eq:sim4}, then $r=a$ and $p=b'$ but there exists no $q$ as there is no rule of type \eqref{eq:sim1} with $(q,b)$.

3) assume $h(s)$ contains symbol $\cv{a}{(q,b)}$, that is, it is of type \eqref{eq:left}  and $h(s)=\ell_d(\cv{a}{(q,b)}yz c)$. It is clear that no letter $r$ with the prefix property exists as the rules of type \eqref{eq:left} have letters $\cv{a}{(q,b)}$ in their image in $h$. 

\end{proof}

The following results are new. They follow from previous lemmata.

\begin{theorem}\label{thm:bounded_delay}
The $\omega$PCP is undecidable for instances of morphisms of bounded delay~2.
\end{theorem}
\begin{proof}
The claim follows directly from Lemmata~\ref{lem:a} and~\ref{lem:b}. 
\end{proof}

It is clear that a non-injective morphism cannot be of bounded delay, and, therefore, a bounded delay morphism is necessarily injective. 

\begin{corollary}
The $\omega$PCP is undecidable for injective morphisms. 
\end{corollary}

\begin{theorem}\label{injPCP_comp}
The $\omega$PCP is $\Pi^0_1$-complete for bounded delay~2 and injective morphisms. 
\end{theorem}

\begin{proof}
$\Pi^0_1$-hardness follows from Theorems~\ref{thm:stcomp} and Lemma~\ref{lem:a}. Also, it can be shown using the construction in~\cite{Finkel} that $\omega$PCP in these special cases is in $\Pi^0_1$. This proves the claim.
\end{proof}

\section{Fixed Shift $\omega$PCP and $\Z$PCP}

In this section we define three new variants of the PCP and give a reasoning for their undecidability. The second and the third one, the shifted $\omega$PCP and $\zps$, will be used in our path for the complexity of the $\Z$PCP. The first variant is presented only for the sake of completeness, as it is the finite version of the shifted $\omega$PCP. 

\begin{prob}[$s$-shift PCP]
Assume $s\in\N$. Given two morphisms $g,h\colon A^*\to B^*$, does there exist a nonempty word $w\in A^*$ and a word $u\in B^*$ such that $|u|=s$ and
\begin{equation}\label{solutionS}
g(w) =u h(w)\, ?
\end{equation}
\end{prob}

Note that the $s$-shift PCP, denoted by $\fsps$, is related to so called    \emph{generalized PCP} (GPCP) where an instance consists of four words $p_1,p_2,s_1,s_2\in \Si^*$ and the set of pairs $\{(u_1,v_1),(u_2,v_2),  \dots,$ $(u_n,v_n)\}$ and $s\in \N$, and the task is to decide whether or not there exists a  sequence $i_1,i_2\dots, i_k$ such that $p_1u_{i_1}u_{i_2}\cdots u_{i_k}s_1 = p_2v_{i_1}v_{i_2}\cdots v_{i_k}s_2$.

\begin{prob}[$s$-shift $\omega$PCP]
Assume $s\in \N$. Given two morphisms $g,h\colon A^*\to B^*$, does there exist an  infinite word $w\in A^\omega$ and a word $u\in B^*$ such that $|u|=s$ and
\begin{equation*}
g(w) = u h(w)  \, ?
\end{equation*}
\end{prob}

Note that the fixed shift $\omega$PCP ($\ops$) can be also presented using the shift in letters: Does there exist an infinite $w\in A^\omega$ such that 
$$
g(w)(i+s)=h(w)(i) \quad \text{ for all }i\in \N. 
$$

Recall our assumption on the indexing letters in the image of a bi-infinite word under morphisms. Let $w\in A^\Z$, and $h\colon A^*\to B^*$ be a morphism. Then 
\[
h(w)(0)= h(w(0))(0),
\]
that is, in the image $h(w)$ the position $0$ is the first letter of the image of the letter $w(0)$. 

\begin{prob}[$s$-shift $\Z$PCP]
Let $s\in \N$. Given two morphisms $g,h\colon A^*\to B^*$,  does there exist an bi-infinite word  $\alpha\in A^\Z$ such that 
\begin{equation}\label{bisolS}
g(\alpha)(i+s) =h(\alpha)(i)
\end{equation}
for all $i\in Z$?
\end{prob}
We shall denote the $s$-shift $\Z$PCP problem by $\zps$.

We will next prove that the $\ops$ is undecidable for any $s\in \N$. Clearly, for $s=0$, the claim follows from the undecidability of the $\omega$PCP. We need to define new pair of morphisms for this problem as for the $g$ and $h$ defined in \eqref{inj_oPCP} for the injective $\omega$PCP, for odd $s=2n+1$ where $n\ge 0$, the words $u=(xd)^nx$ and $w=x^\omega$ for $x\in\Delta_1$ are a solution of the $\ops$, as $g(w)=(xd)^\omega = (xd)^nx(dx)^\omega=uh(w)$.   

\begin{theorem}\label{thm:opcps}
For all $s\in\N$, the $\ops$ is undecidable.
\end{theorem}

\begin{proof}
Assume $s\in \N$ and let the morphisms $g,h\colon A^*\to
B^*$ be as in ~\eqref{inj_oPCP}. Let $e$ be a new letter and define the morphisms $g',h'\colon A^*\to (B\cup\{e\})^*$ 
so that 
\begin{equation*}
h'(d)=er_{e^{s}}(h(d))=e r_{e^{s}}(\ell_{d}(\#\cv{(q_0,\B)}{S}\$)),\quad  g'(d)=e^{s+1}r_{e^{s}}(d)=e^{s+1}de^{s},\\ 
\end{equation*}
and for all $y\in \Delta_1\cup R$, $y\ne d$, set $h'(y)=r_{e^{s}}(h(y))$ and $g'(y)=r_{e^{s}}(g(y))$. 

Firstly, if there is a solution $u\in A^s$ and $w\in A^\omega$ to $(g',h')$ as an instance of $\ops$, then $w=dw'$ for some $w'\in A^\omega$. This follows from the fact that for all $a\in A$, $a\ne d$, $g'(a)(s)=e$, but $h'(a)(0)\in B$. Now $g'(d)$ is comparable with $e^{s} h'(d)$ implying that $u=e^s$. Based on a previous sections, it is obvious (modulo $r_{e^s}$) that there exists a required infinite word $w$ for the solution of $\ops$ if and only if TM $\M$ does not halt on the empty input. Thus the solution constructed from the infinite derivation of $S_\M$ is the
only possible solution compatible with the forced initial letter $d$.
\end{proof}

Moreover, in the case where such a solution exists, the above argument
shows that the solution must begin with \(d\). Thereafter the equality \(g'(w)=e^s h'(w)\), together with the determinism of \(M\) and hence the
\(\Theta\)-determinism of \(S_M\), forces at each step the unique next rule
of \(S_M\) applicable to the current configuration. Thus the solution
constructed from the infinite derivation of \(S_M\) is the only possible
solution of this form.

Similarly to Theorem~\ref{injPCP_comp} we can prove the following

\begin{theorem}
For all $s\in\N$, the $\ops$ is $\Pi^0_1$-complete.  
\end{theorem}

Now the undecidability of the $\zps$ follows  from Theorem~\ref{thm:opcps}. 

\begin{theorem}\label{thm:zpcps}
For all $s\in\N$, the $\zps$ is undecidable.
\end{theorem}

\begin{proof}
Let the morphisms $g$, $h$, $g'$ and $h'$ be as in the proof of Theorem~\ref{thm:opcps}. Let $A'$ be a copy of the alphabet $A$, that is,  $A'=\{x'\mid x\in A\}$. Define the morphisms $g_1$ and $h_1$, so that for all $y\in A$, 
$g_1(y)=g'(y)$ and $h_1(y)=h'(y)$. Then, let $f$ be a new letter, and define 
$$
h_1(d') =f^sdfe^s , \quad g_1(d') = \ell_{f^s} (h(d)^R)f = \ell_{f^{s}}(r_{d}(\$\cv{(q_0,\B)}{S}\#))f
$$
and for $y\in A$, $y\ne d$ define
$$
g_1(y') =\ell_{f^s}(h(y)^R), \quad h_1(y') = \ell_{f^s} (g(y)^R).
$$
The idea here is that to the right the solution of the instance $(g_1,h_1)$ of the 
$\zps$ is a solution $w$ of the instance $(g',h')$ of $\ops$ and to left $(w')^R$ where $w'$ is the primed copy of the word $w$ (that is, for all $i\in \N$, $w'(i)=w(i)'$). The claim follows from the following two observations: 
\begin{enumerate}
    \item No image of letter from $A$ is $s$-shift comparable with an image of letter from $A'$ (and vice versa) as the $e\ne f$. 
    \item By the proof of Theorem~\ref{thm:opcps}, the only possible $s$-shift to the right begins with $d$, and similarly to the left with $d'$.
 \end{enumerate}   
    Therefore, 
$$
g_1(w'^R) g_1(w)=h_1(w'^R) h_1(w)
$$
where \(w \in A^\omega\) is the solution forced by the unique infinite
derivation of \(S_M\) from the initial word
\(\#\cv{(q_0,\B)}{S} \$\) and $w'\in (A')^\omega$ is the primed copy of $w$.  
\end{proof}

In~\cite{FHHS} it was proved that the $\zps$ is in $\Pi^0_1$ and as $\Pi^0_1$-hardness of $\zps$ follows from the above, that is, the original TM $\M$ is non-halting if and only if the instance $(g_1,h_1)$ of the 
$\zps$ has a solution, we have that 

\begin{theorem}
For all $s\in\N$, the $\zps$ is $\Pi^0_1$-complete.  
\end{theorem}

Unfortunately, the rather trivial proofs of this section do not help us on proving that the $\Z$PCP is $\Pi^0_1$-hard. For that we need to use the reversibility property of the constructed semi-Thue system $S_\M$. Indeed, the instance $(g_1,h_1)$ described above has a trivial solution $x^\omega$ for $x\in \Delta_1$ to the $Z$PCP. Such solutions need to be eliminated by forcing the solution to contain a letter from $R$ together with symbols $\#$ and $\$$, to force the solution to follow the derivation of $S_\M$. Then the reversibility will force the used computation to be from the initial configuration $\#\cv{(q_0,\B)}{S}\$ $.

\section{Another construction for the $\Z$PCP}\label{finalred}

We return to the deterministic and reversible semi-Thue system $S_\M$ defined in Sec.~\ref{sec:revsT}. Recall the definition of $\Lambda$, that is 
\begin{equation*}
\begin{split}
&\Lambda=\Delta \cup \Theta\cup\set{\#,\$ }, \quad\text{ where }\\    
 &\Delta = \Gamma \cup (Q \times \Gamma) \text{ and } \\
 &\Theta = \Delta \times (\{ S \} \cup (Q \times \Gamma)).
\end{split}
\end{equation*}
First, we take two copies of alphabet $\Delta$, that is, $\Delta_i=\{x_i \mid x\in \Delta\}$ for $i=1,2$. Then consider the alphabet 
$$
C=\Delta_1\cup \Delta_2\cup \Theta\cup \set{\#,\$} \cup R,
$$
and take an (overlined) copy of $C$, that is, $\ov{C}=\set{\ov{x}\mid x\in C}$. 
The idea behind the alphabets is the following: a configuration 
$$
(q_0,\B)\cdots (q_m,a_m)\# b_1\cdots b_k \cv{X}{Y} b_{k+1}\cdots b_n\$
$$
of the semi-Thue system $S_\M$  will appear in the images of the morphisms in the form (excluding the desynchronizations at this point) 
\begin{equation}\label{form_nodes}
(q_0,\B)_1\cdots (q_m,a_m)_1\#(b_1)_1\cdots (b_k)_1 \cv{X}{Y} (b_{k+1})_2\cdots (b_n)_2\$,\end{equation}
i.e., symbols from $\Delta_1$ are printed only before a symbol $\cv{X}{Y}\in \Theta$ and after the symbol from $\Theta$ only symbols from $\Delta_2$ are printed. The correct appearance of symbols in~\eqref{form_nodes} is forced using $d^4$ and $f^4$ as desyncronization patterns. In the images of morphism $h$ the words $d^4$ and $f^4$ are forced between the letters from  $C$ and $\overline{C}$, respectively. Note also that using a standard technique here in $h$, the images of non-overlined letters are overlined and images of overlined are non-overlined, whereas in $g$ the overlined are mapped to overlined letters and non-overlined to non-overlined letters. This forces the overlined and non-overlined configurations of the form~\eqref{form_nodes} alternate in the image of any (infinite) solution. 

Now the desyncrononization in $g$ is built in the following way (see Table~\ref{table:morph}): let $k$ be the desyncronizing symbol (that is, $k\in \{d,f\}$). Assume that we need to print a desyncronized word of a configuration of the form~\eqref{form_nodes} with images of $g$ having $k^4$ between the symbols of $C$ (or $\ov{C}$). Indeed, we need to print symbols $\#$ (or $\ov{\#}$), $\$$ (or $\ov{\$}$) and a symbol from $\Theta$ (or $\ov{\Theta}$) with $k^4$ desynchronization. Assume that before the part we are now trying to cover with $g$, a similar configuration with desyncronization is just printed by an image of $g$. Therefore, we may assume that a single $k$ from image of $d$ (the initial configuration desyncronized) or from image of $\$$ or $\ov{\$}$ is appearing in the image. The image of the form~\eqref{form_nodes} is split into parts: 
\begin{itemize}
    \item before symbol $\#$ or $\ov{\#}$:  we need to get $k^3$ to match the single $k$, letters $x_1\in \Delta_1$ such that $x\in Q\times \Gamma$  are used implying again single $k$ to the end. To eventually have a configuration of the form~\eqref{form_nodes} printed as an image, at some point we need to have symbol $\#$ or $\ov{\#}$ to the image.
    \item Symbol $\#$ or $\ov{\#}$ is printed: if it is printed as an image of  
    a symbol $\#$ or $\ov{\#}$, and then after it the image ends with $k^2$. Then symbols $x_1\in \Delta_1$ with $x\in \Delta\setminus (Q\times \Gamma)$ are necessarily used to get $k^4$ between the letters. Anyway, at some point we need have image of a letter in $R$ or $\ov{R}$ to reach $\$$ or $\ov{\$}$ eventually. This implies that also a unique symbol from $\Theta$ or $\ov{\Theta}$ will appear in the image. Note also that an image of a symbol from $R$ or $\ov{R}$ is also another possibility to get $\#$ or $\ov{\#}$ to the image. 
    \item Based on the above, a unique symbol from $\Theta$ or $\ov{\Theta}$ appears in the image as an  image of a symbol from $R$ or $\ov{R}$. Such a unique symbol is necessarily printed in order to reach $\$$ or $\ov{\$}$ with $k^4$ desynchronization, as the $k^2$ overflow we have before the symbol must be changed to $k^3$. Assume now that $k^3$ appears as an overflow in the end of the image. 
    \item Based on the $k^3$ overflow, we need a letter whose image has a single $k$ as a prefix. Therefore, between a symbol from $\Theta$ or $\ov{\Theta}$ and $\$$ or $\ov{\$}$ only letters in $ \Delta_2$ can be used to have $k^4$ between the letters, or  
    \item symbol $\$$ or $\ov{\$}$ is printed and at the end in order to cover a configuration of the form~\eqref{form_nodes}, and we have again single $k$ as an overflow.
\end{itemize}

Here is the detailed definition of the morphisms 
$$
g,h\colon C\cup \ov{C}\cup \set{d} \to (C\setminus R)\cup (\ov{C}\setminus \ov{R}) \cup\set{e,f},
$$
where $\ov{X}$ denotes the overlined copy of the alphabet $X$. 
\begin{table}[ht]
\noindent \centering \caption{Morphisms $g$ and $h$: In the right half of the table, the displayed formulae are the overlined
copies of the corresponding formulae in the left half, with $d$ and $f$
interchanged as indicated.} \label{table:morph}
\scalebox{0.72}{
\hspace{-1.2cm}\begin{tabular}{|c|c|c|c|c|}    \hline
 $x$ &   $g(x)$ & $h(x)$ &$g(\ov{x})$ & $h(\ov{x})$ \\   \hline &&&& \\ [-1em]
 $x=(q,b)_1$, $(q,b)\in (Q\times \Gamma)$  & $ dddxd$ & $\ell_{f^4}(\ov{x})$ &$fff\ov{x}f $& $\ell_{d^4}(x)$\\   \hline &&&& \\ [-1em]
  $x=\#$ & $ ddd\#dd$ & $\ell_{f^4}(\ov{\#})$ &$fff\ov{\#}ff$ & $\ell_{d^4}(\#)$\\    \hline &&&& \\ [-1em]
 $x=x_1$, $x\in \Delta\setminus (Q\times \Gamma)$  & $ ddx_1dd$ & $\ell_{f^4}(\ov{x_1})$ &$ff\ov{x_1}ff $& $\ell_{d^4}(x_1)$\\   \hline
 $x= a\cv{b}{S} z \derive{}  a b\cv{z}{S}$ in \eqref{eq:right}&   $dd r_{d^4}\left(a_1 \cv{b}{S}\right) z_2ddd$       & $\ell_{f^4}\left(\ov{a_1 b_1 \cv{z}{S}}\right)$ & $ff r_{f^4}\left(\ov{a_1 \cv{b}{S}}\right) \ov{z_2} fff$       & $\ell_d\left(a_1 b_1 \cv{z}{S}\right)$  \\ \hline
$  x= a \cv{(q,b)}{S}c\derive{} \cv{(q',a)}{(q,b)} b'c$ in \eqref{eq:sim1} &   $ddr_{d^4}\left(a_1 \cv{(q,b)}{S}\right)c_2 ddd$   & $\ell_{f^4}\left(\ov{\cv{(q',a)}{(q,b)} b'_2c_2}\right)$        & $ffr_{f^4}\left(\ov{a_1 \cv{(q,b)}{S}}\right)\ov{c_2} fff$   & $\ell_{d^4}\left(\cv{(q',a)}{(q,b)} b'_2c_2\right)$ \\ \hline
$x= \# \cv{(q,b)}{S}c\derive{} \# \cv{(q',\B)}{(q,b)} b'c $ in \eqref{eq:sim2} &  $ddd r_{d^4}\left(\# \cv{(q,b)}{S}\right)c_2 ddd$  & $\ell_{f^4}\left(\ov{\# \cv{(q',\B)}{(q,b)} b_2'c_2} \right)$  & $fff r_{f^4}\left(\ov{\# \cv{(q,b)}{S}}\right)\ov{c_2} fff$  & $\ell_{d^4}\left(\# \cv{(q',\B)}{(q,b)} b_2'c_2 \right)$ \\    \hline
 $x=a\cv{(q,b)}{S} c\derive{}  a b' \cv{(q',c)}{(q,b)}$ in \eqref{eq:sim3} & $ddr_{d^4}\left(a_1\cv{(q,b)}{S}\right)c_2 ddd$   & $\ell_{f^4}\left(\ov{a_1 b'_1 \cv{(q',c)}{(q,b)}}\right)$ &$ffr_{f^4}\left(\ov{a_1\cv{(q,b)}{S}}\right)\ov{c_2} fff$   & $\ell_{d^4}\left(a_1 b_1' \cv{(q',c)}{(q,b)}\right)$   \\ \hline
$x=a\cv{(q,b)}{S} \$\derive{} a b' \cv{(q',\B)}{(q,b)} \$ $ in \eqref{eq:sim4} &    $ddr_{d^4}\left(a_1\cv{(q,b)}{S}\right) \$ f $ & $\ell_{f^4}\left(\ov{a_1 b_1' \cv{(q',\B)}{(q,b)} \$} \right)$ &$ffr_{d^4}\left(\ov{a_1\cv{(q,b)}{S}}\right) \ov{\$} d $ & $\ell_{d^4}\left(a_1 b_1' \cv{(q',\B)}{(q,b)} \$ \right)$  \\ \hline
$x=a \cv{y}{(q,b)}z\derive{}  \cv{a}{(q,b)}yz$ in \eqref{eq:left} & $ddr_{d^4}\left(a_1 \cv{y}{(q,b)}\right) z_2ddd$   & $\ell_{f^4}\left(\ov{\cv{a}{(q,b)}y_2z_2}\right)$ &$ffr_{f^4}\left(\ov{a_1 \cv{y}{(q,b)}}\right) \ov{z_2}fff$   & $\ell_{d^4}\left(\cv{a}{(q,b)}y_2z_2\right)$ \\
    \hline
$x=\# \cv{y}{(q,b)}z\derive{} (q,b) \# \cv{y}{S}z$ in \eqref{eq:forget}  & $dddr_{d^4}\left(\# \cv{y}{(q,b)}\right) z_2ddd$  & $\ell_{f^4}\left(\ov{(q,b) \# \cv{y}{S}z_2}\right) $ & $fffr_{f^4}\left(\ov{\# \cv{y}{(q,b)}}\right) \ov{z_2}fff$  & $\ell_{d^4}\left((q,b) \# \cv{y}{S}z_2\right) $   \\
    \hline &&&& \\ [-1em]
$x=x_2$, $x\in \Delta$& $ dx_2ddd$ & $\ell_{f^4}(\ov{x_2})$ &$f\ov{x_2}fff $& $\ell_{d^4}(x_2)$\\   \hline  &&&& \\ [-1em]
 $x=\$ $ & $d\$ f $ &  $\ell_{f^4}(\ov{\$})$ & $f\ov{\$}d$ &  $\ell_{d^4}(\$)$ \\ \hline 
$x=d$ & $ed$ & $ e \ell_{d^4}\left(\#\cv{(q_0,\B)}{S}\$\right) $ & &  \\ \hline
\end{tabular}}
\end{table}

To recall the idea shortly:  the construction is similar to the construction in \eqref {inj_oPCP} except that the desynchronoization is done by $d^4$ and $f^4$ which gives tools to force the symbols $\#$, $\cv{X}{Y}\in \Theta$ and $\$$ to appear exactly once, in the correct order, in one configuration in the image. The overlining on the other hand forces that in a solution to PCP or $\Z$PCP both overlined and non-overlined letters must appear, and the overlined and non-overlined configurations alternate in the image.  

It is rather straightforward to prove the following variants of lemmata~\ref{lem:sim} and \ref{lem:a} for these new morphisms $g$ and $h$. We shall just outline the differences implied by the new methods in the construction.

For $j=1,2$, define $(\cdot)_j\colon\Delta\cup\{\#,\$\}\to
\Delta_j\cup\{\#,\$\}$ by
\[
(x)_j=
\begin{cases}
x_j, & \text{if } x\in\Delta,\\
x, & \text{if } x\in\{\#,\$\}.
\end{cases}
\]

The following lemma follows directly from the definitions of $g$ and $h$.

\begin{lemma} \label{lem:simrev}
Let $\beta_0=\#\cv{(q_0,\B)}{S}\$\derive{} \beta_1 \derive{} \beta_2 \derive{} \ldots$ be a
derivation in $S_{\M}$, where 
$$ \beta_j=w'uw''=x\#yXz\$ \quad \text{ and }\quad  \beta_{j+1}=w'vw''=w\#qYp\$ 
$$ 
with  $t=(u\derive{} v)\in R$, for some $x,y,z,w,q,p \in \Delta^*$ for $j=0,1,2, \dots $, and $X,Y\in \Theta$. 
Now, for $j=0$
\[
g(d) =ed \quad \text{ and }\quad h(d) = e\ell_{d^4}(\beta_0),
\]
and for $j=1,2,\dots$, with $\beta_j=w'uw'' $  
\begin{align*}
&g((w')_1t(w'')_2) = ddd r_{d^4}\left((x\#y)_1 X (z)_2\right)\$f, \\ 
&h((w')_1t(w'')_2) = \ell_{f^4}\left(\ov{ (w\#q)_1Y(p)_2\$} \right) , \\
&g \left(\ov{(w')_1t(w'')_2}\right) = fff r_{f^4}\left(\ov{(x\#y)_1 X (z)_2}\right)\ov{\$}d, \quad \text{ and }\\
&h\left(\ov{(w')_1t(w'')_2}\right) = \ell_{d^4}\left( (w\#q)_1Y(p)_2\$\right).  
\end{align*}
\end{lemma}

\begin{lemma}\label{lem:gcover}
If $$ddd r_{d^4}(u)\$ \leq_p g(w)$$ for some nonempty word $u\in \left(\left(C\cup \ov{C}\cup  \set{d}\right)\setminus \set{\$,\ov{\$}} \right)^*$, then there exists $w'\in (Q\times\Gamma)^*\# (\Delta\setminus (Q\times\Gamma))^*$ or $w'\in (Q\times\Gamma)^*$, and $w''\in \Delta^*$ or $w''\in \Delta^*\$ $ such that,
\begin{enumerate}
    \item $ddd r_{d^4}(u)\$f=g((w')_1t(w'')_2) $, and  
    \item $(w')_1t(w'')_2 \leq_p w$.
\end{enumerate}
\end{lemma}

\begin{proof}
Assume that $u\in \left(\left(C\cup \set{d}\right)\setminus \set{\$} \right)^*$, $u\ne \varepsilon$ and that $ddd r_{d^4}(u)\$f \leq_p g(w)$. Now $d^4\$ f \le_s ddd r_{d^4}(u)\$ f $. There are two ways to cover the $\$ $ in $g(w)$, either (a) by image of $t\in R$ of type~\eqref{eq:sim4} or (b) by an image of $\$ $. 

Let us consider the first case (a). Then 
$$
ddd r_{d^4}(u)\$ f = ddd r_{d^4}(u') d^{-2} \cdot dd g(t)
$$ 
implying that to match the $ddd$ in the beginning, $u'$ must be of the form $(u'')_1\# (u''')_1$, where $u''\in (Q\times \Gamma)^*$ and $u'''\in (\Delta\setminus(Q\times \Gamma))^*$. Now setting $w'=u'$ and $w''=\varepsilon$ gives that $g((w')_1t(w'')_2)= ddd r_{d^4}(u)\$f $ 

For the case (b), $u$ has a suffix $ z \in \Delta_2^*$, and 
$$
ddd r_{d^4}(u)\$ f = ddd r_{d^4}(u')d^{-1} \cdot d r_{d^4}(z)\$ f.
$$
Assume that $z$ is the maximal suffix from $\Delta_2^*$ of $u$.  To match the prefix $ddd$, the word $u'$ must have a letter $t\in R$ as a suffix ($t$ is not of type~\eqref{eq:sim4}). We have again two cases, either 
\begin{enumerate}
    \item[(i)] $t$ is of the form~\eqref{eq:sim2} or~\eqref{eq:forget}, or 
    \item[(ii)] $t$ is of the form~\eqref{eq:right}, \eqref{eq:sim1}, \eqref{eq:sim3} or~\eqref{eq:left}.     
\end{enumerate}
For case (i), $ddd r_{d^4}(u)\$ f = ddd r_{d^4}(u'')d^{-3} \cdot g(t)\cdot dr_{d^4}(z)\$ f$
for some $u''$. Necessarily $u''=(x)_1$ where $x\in (Q\times \Gamma)^*$. Now setting 
$w'=x$ and $w''=z'\$$ where $z=(z')_2$, we have that $ddd r_{d^4}(u)\$f=g((w')_1t(w'')_2) $. 

Finally, for case (ii), $ddd r_{d^4}(u)\$ f = ddd r_{d^4}(u'')d^{-2} \cdot g(t)\cdot dr_{d^4}(z)\$ f$
for some $u''$. Necessarily $u''=(x)_1\# (y)_1$ where $x\in (Q\times \Gamma)^*$ and $y\in (\Delta_1\setminus (Q\times \Gamma))^*$. Now setting 
$w'=x\#y$ and $w''=z'\$$ where $z=(z')_2$, we have that $ddd r_{d^4}(u)\$f=g((w')_1t(w'')_2) $. 

Claim 2 follows from the above and the uniqueness of suitable rule $t$. 
\end{proof}

The following lemma is symmetric with the previous one. 

\begin{lemma}\label{lem:gcover2}
If $$fff r_{f^4}(u)\ov{\$} \leq_p g(w)$$ for some nonempty word $u\in \left(\left(C\cup \ov{C}\cup  \set{d}\right)\setminus \set{\$,\ov{\$}} \right)^*$, then exists $w'\in (Q\times\Gamma)^*\# (\Delta\setminus(Q\times\Gamma))^*$ or $w'\in (Q\times\Gamma)^*$, and $w''\in \Delta^*$ or $w''\in \Delta^*\$ $ such that,
\begin{enumerate}
    \item $fff r_{f^4}(u)\ov{\$}d=g(\ov{(w')_1t(w'')_2}) $, and  
    \item $\ov{(w')_1t(w'')_2 }\leq_p w$.
\end{enumerate}
\end{lemma}

\begin{lemma} \label{lem:arev}
Let $\M$ be a Turing machine and $S_{\M}$ be the corresponding semi-Thue
system. Then $S_\M$ is non-terminating for the initial word
 $\#\cv{(q_0,\B)}{S}\$ $ if and only if there exists an infinite word $w\in \left( C\cup \ov{C}\cup \set{d}\right)^\omega$ such that
 $g(w)=h(w)$.
\end{lemma}

\begin{proof}
Assume first that $S_{\M}$ is not terminating for $\#\cv{(q_0,\B)}{S}\$$, that is, there exists 
$$
\beta_0=\#\cv{(q_0,\B)}{S}\$ \derive{}\beta_1\derive{} \beta_2 \derive{} \ldots\,,
$$
where  $\beta_j=x_j u_jX_ju'_j y_j$ and $\beta_{j+1}=x_j
v_jY_jv'_j y_j=x_{j+1}u_{j+1}X_{j+1}u'_{j+1}y_{j+1}$ with $t_j=(u_jX_ju'_j\derive{} v_jY_jv'_j)\in R$ being the rule
applied in the derivation step  $\beta_j\derive{} \beta_{j+1}$ with $X_j,Y_j\in \Theta$. Now, for
$$
w=d(x_0)_1t_0(y_0)_2\ov{(x_1)_1t_1(y_1)_2}(x_2)_1t_2(y_2)_2\ov{(x_3)_1t_3(y_3)_2}\cdots,
$$
we have by Lemma~\ref{lem:simrev} that
\begin{equation*}
\begin{split}
&g(w)\\
&=ed\cdot ddd r_{d^4}\left((x_0u_0)_1X_0(u'_0y_0)_2\right)\$ f\cdot fff r_{f^4}\left(\ov{(x_1u_1)_1X_1(u'_1y_1)_2}\right)\ov{\$ } d \cdot\\
& ddd r_{d^4}\left(x_2u_2)_1X_2(u'_2y_2)_2\right)\$ f\cdot fff r_{f^4}\left(\ov{(x_3u_3)_1X_3(u'_3y_3)_2}\right)\$d\cdots   \\
&= e\ell_{d^4}\left(\#\cv{(q_0,\B)}{S}\$\right) \ell_{f^4}\left(\ov{(x_1u_1)_1X_1(u'_1y_1)_2\$}\right) \cdot\\  &\ell_{d^4}\left((x_2u_2)_1X_2(u'_2y_2)_2\$\right)\ell_{f^4}\left(\ov{(x_3u_3)_1X_3(u'_3y_3)_2} \right) \cdots\\
&=e\ell_{d^4}\left(\#\cv{(q_0,\B)}{S}\$\right) \ell_{f^4}\left(\ov{(x_0v_0)_1Y_0(v'_0y_0)_2\$}\right)\cdot \\ & \ell_{d^4}\left(x_1v_1)_1Y_1(v'_1y_1)_2\$ \right)\ell_{f^4}\left(\ov{(x_2v_2)_1Y_2(v'_2y_2)_2\$} \right) \cdots \\
&=h(w)\,,
\end{split}
\end{equation*}
and hence $g(w)=h(w)$ as required.

Assume next that there is a  word $w\in \left( C\cup \ov{C}\cup \set{d}\right)^\omega$ such that $g(w)=h(w)$ and, to the  contrary, that $S_\M$ terminates on $\#\cv{(q_0,\B)}{S}\$ $.

Necessarily, $w=dw'$ for some word $w'\in \left( C\cup \ov{C}\cup \set{d}\right)^\omega$, since only $g(d)$ and $h(d)$ are prefix comparable. Now,
$$h(d)=e\ell_{d^4}(\#\cv{(q_0,\B)}{S}\$)=e\ell_{d^4}(\beta_0)\,,$$
and therefore, $ddd r{d^4}(\#\cv{(q_0,\B)}{S})\$\le_p g(d)^{-1}g(w)$. As the only ways to have $\$$ as an image of $g$ have suffix $\$ f$, it means that actually 
$ddd r{d^4}(\#\cv{(q_0,\B)}{S})\$f\le_p g(d)^{-1}g(w)$. By Lemma~\ref{lem:gcover}, there exist $x_0$, $y_0$ and $t_0$ such that $(x_0)_1t_0(y_0)_2\le_p d^{-1}w$, where $\beta_0=x_0 u_0 y_0$ and $\beta_{1}=x_0 v_0 y_0$ and, moreover, $t_0=(u_0\derive{} v_0)\in R$ is the rule applied in the first step $\beta_0\derive{} \beta_{1}$. Also,
$h((x_0)_1t_0(y_0)_2)=\ell_{f^4}\left(\ov{(x'\# x'')_1Y(y')_2\$} \right)$, where $x'\#x''Yy'\$=\beta_1$. Now $fff r_{f^4}\left(\ov{(x'\# x'')_1Y(y')_2}\right)\ov{\$} \le g(d(x_0)_1t_0(y_0)_2)^{-1}g(w)$ and by Lemma~\ref{lem:gcover2}, 
exists $x_1$, $y_1$ and $t_1$ such that $(x_1)_1t_1(y_1)_2\le_p (d(x_0)_1t_0(y_0)_2)^{-1}w$, where $\beta_1=x_1 u_1 y_1$ and $\beta_{2}=x_1 v_1 y_1$ and, moreover, $t_1=(u_1\derive{} v_1)\in R$ is the rule applied in the second step $\beta_0\derive{} \beta_{1}\derive{} \beta_{2}$.

By applying Lemmas~\ref{lem:gcover} and \ref{lem:gcover2} iteratively, we obtain
a derivation
\[
\beta_0=\#\cv{(q_0,\B)}{S}\$ \derive{}\beta_1\derive{}^* \beta_n
\]
in ${S_{\M}}$. Let $\beta_n=x_n\cv{(\halt,a)}{S} y_n\$ $. Based on the above, at some  point in the image $h(w)$ we have either  
 \[
\ell_{d^4}\left((x_n)_1\cv{(\halt,a)}{S} (y_n)_2\$\right) \quad\text{ or }\quad
\ell_{f^4}\left(\ov{(x_n)_1\cv{(\halt,a)}{S} (y_n)_2\$}\right).
\]
Now the letter $\cv{(\halt,a)}{S}$ or $\ov{\cv{(\halt,a)}{S}}$ do not appear in images of $g$, so there is no way to cover the above words with images of $g$, a contradiction.
\end{proof}

The next lemma is the key step in the proof of $\Pi^0_1$-hardness of ZPCP. It follows from our desynchronization with $d^4$ and $f^4$, overlined symbols and, finally, from the reversibility of $S_\M$.

\begin{lemma}\label{no_biinf_sol}
The instance $(g,h)$ does not have a solution as an instance of $\Z$PCP. 
\end{lemma}

\begin{proof}
Assume, to the contrary, that there is a word
$\alpha \in (C\cup \overline C\cup\{d\})^\mathbb{Z}$ and a shift
$\tau \in \mathbb{Z}$ such that
$g(\alpha)(i+\tau)=h(\alpha)(i)$ for all $i\in\mathbb{Z}$.

We first show that the letter $d$ cannot occur in $\alpha$. Indeed, the
letter $e$ occurs only in the images of $d$, and in both $g(d)$ and $h(d)$
it occurs as the first letter. Hence any occurrence of $d$ in $\alpha$ gives,
after aligning the corresponding occurrences of $e$, a right-infinite suffix
which must satisfy the same forced comparison as in the proof of
Lemma~\ref{lem:arev}. Thus the right-hand side of such an occurrence of $d$ is forced, exactly as
in the proof of Lemm~\ref{lem:arev}, to follow the unique derivation of $S_M$ starting
from $\#\cv{(q_0,\star)}{S}\$$. In particular, no further
occurrence of $d$ appears on this right-hand side. Looking to the left of the
aligned occurrence of $e$, the images of the letter immediately preceding
$d$ would have to be suffix-comparable under $g$ and $h$. Indeed, no image ending immediately before an occurrence of $e$ can have
a suffix matching the corresponding suffix of the other morphism, since
the only occurrences of $e$ are the initial letters of $g(d)$ and $h(d)$,
and the preceding images end in letters from
$(C\setminus R)\cup(\overline C\setminus \overline R)\cup\{d,f\}$. A direct inspection
of Table~\ref{table:morph} shows that this never happens for any letter of
$C\cup\overline C\cup\{d\}$. This contradiction proves that $d$ does not
occur in $\alpha$.

Thus $\alpha\in (C\cup\overline C)^\mathbb{Z}$. We record the two
consequences of Table~1 that will be used below. First, the morphism $g$
preserves the overlining of letters, whereas $h$ changes overlined letters
to non-overlined ones and non-overlined letters to overlined ones. Therefore
a one-sided infinite factor entirely over $C$ or entirely over $\overline C$
cannot occur in a solution. Hence $\alpha$ can be written in the form
\[
   \alpha=\cdots \delta_{-1}\gamma_{-1}\delta_0\gamma_0
              \delta_1\gamma_1\cdots ,
\]
where, for every $i\in\mathbb{Z}$, $\delta_i\in \overline C^+$ and
$\gamma_i\in C^+$, or with the roles of $C$ and $\overline C$ interchanged.

Second, the only way, in an image under $g$, to pass from a block
desynchronized by $d^4$ to a block desynchronized by $f^4$ is through an
image ending in $\$f$; similarly, the only way to pass from a block
desynchronized by $f^4$ to a block desynchronized by $d^4$ is through an
image ending in $\$d$. Consequently, whenever an image under $h$ changes
from an $f^4$-desynchronized block to a $d^4$-desynchronized block, or
conversely, the corresponding image under $g$ must contain the final marker
$\$$ of a simulated configuration.

Consider now a factor
\[
   h(\gamma_i\delta_{i+1}\gamma_{i+1})
      = \ell_{f^4}(u_i)\ell_{d^4}(v_{i+1})\ell_{f^4}(u_{i+1})
\]
of $h(\alpha)$, where $u_i,v_{i+1},u_{i+1}$ are nonempty words over the
corresponding non-rule alphabets. By the previous paragraph, the matching
factor of $g(\alpha)$ contains a factor of the form
$ddd r_{d^4}(u)\$f$ or of the form $fff r_{f^4}(u)\$d$. Applying
Lemma~8 or Lemma~9, respectively, we obtain a factor
$(w')_1t(w'')_2$ of $\alpha$, with $t\in R$, such that its image under $g$
is exactly the encoded form of one configuration of $S_M$, and its image
under $h$ is the encoded form of its successor configuration.

Since $\alpha$ is bi-infinite, the same argument can be applied repeatedly
to the left. Thus, from one encoded configuration
$x\#yXz\$\in (Q\times\Gamma)^*\#\Delta^*\Theta\Delta^*\$$,
we obtain an infinite sequence of reverse derivation steps in $S_M$:
\[
   \cdots \longrightarrow \beta_{-2}
          \longrightarrow \beta_{-1}
          \longrightarrow \beta_0=x\#yXz\$ .
\]
This contradicts Lemma~3, which states that no infinite reverse derivation
can start from such a word. Therefore the assumed bi-infinite solution
$\alpha$ cannot exist.
\end{proof}

Using the idea of the proof of Theorem~\ref{thm:zpcps}, that is, by taking a reverse copy of the morphisms $g$ and $h$ with totally new alphabets and then combining them with $g$ and $h$ to get morphisms $g_1$ and $h_1$, Lemma~\ref{no_biinf_sol} directly implies the following lemma. More precisely, take a disjoint copy $D$ of the alphabet
$C\cup\overline C\cup\{d\}$ and define on $D$ a reversed copy of the
morphisms $g$ and $h$, replacing each image by its reversal and using
fresh desynchronizing letters. The two copies are joined only through the
forced initial markers corresponding to the letter $d$. By Lemma~\ref{no_biinf_sol}, no
bi-infinite solution can be contained entirely in one copy. Hence any
solution of the combined instance must contain the joining point, and its
right-hand side is exactly the infinite computation produced in Lemma~10,
while its left-hand side is the reversed copy of the same computation.

\begin{lemma}\label{lem:12}
Let $\M$ be a Turing machine and $S_{\M}$ be the corresponding semi-Thue
system. Then $S_\M$ is non-terminating for the initial word
 $\#\cv{(q_0,\B)}{S}\$ $ if and only if there exists a bi-infinite word $\alpha$ such that
 $g_1(\alpha)=h_1(\alpha)$.
\end{lemma}

We now state the main result of this section. 

\begin{theorem}
The $\Z$PCP is $\Pi_1^0$-hard. Moreover, $\Z\text{PCP}\in \Sigma^0_2 \setminus (\Pi^0_1 \cup \Sigma^0_1)$. 
\end{theorem}

\begin{proof}
The first claim follows from Lemma~\ref{lem:12}. 

By \cite{FHHS}, we have
\(\Z\text{PCP}\in\Sigma^0_2\) and
\(\Z\text{PCP}\notin\Pi^0_1\), where the latter non-membership result is
based on the construction of \cite{bi-inf}. Moreover, the
\(\Pi^0_1\)-hardness proved above implies that
\(\Z\text{PCP}\notin\Sigma^0_1\). Hence
\[
   \Z\text{PCP}\in\Sigma^0_2\setminus(\Pi^0_1\cup\Sigma^0_1).
\] 
\end{proof}

\section{Conclusion}

We have determined the exact arithmetical complexity of several infinite variants of the Post Correspondence Problem. In particular, we proved that the $\omega$PCP remains $\Pi^0_1$-complete even for injective morphisms, indeed for morphisms of bounded delay $2$, that $\ops$ and $\zps$ are $\Pi^0_1$-complete for every $s\in\N$. 

For the ordinary $\Z$PCP, we proved that it is $\Pi^0_1$-hard and belongs to
$\Sigma^0_2\setminus(\Pi^0_1\cup\Sigma^0_1)$. Thus $\Z$PCP lies properly at
the second level of the arithmetical hierarchy, although the question whether
it is $\Sigma^0_2$-complete remains open.

\paragraph{Acknowledgements.} During the preparation of this manuscript, the authors used ChatGPT
(OpenAI) to assist with language editing, proofreading, and improving the
presentation. All mathematical arguments, results, and final editorial
decisions are the responsibility of the authors.

\bibliographystyle{abbrv}
\bibliography{database.bib}

\end{document}